\documentclass[sn-mathphys,Numbered]{sn-jnl}

\usepackage{graphicx}%
\usepackage{multirow}%
\usepackage{amsmath,amssymb,amsfonts}%
\usepackage{amsthm}%
\usepackage{mathrsfs}%
\usepackage[title]{appendix}%
\usepackage{xcolor}%
\usepackage{textcomp}%
\usepackage{manyfoot}%
\usepackage{booktabs}%
\usepackage{algpseudocode}%
\usepackage{listings}%
\usepackage{tikz}
\usepackage{float}
\usepackage{algorithm2e}
\usepackage{ragged2e}
\usepackage{natbib}
\newtheorem{definition}{Definition}%
 \newtheorem{theorem}{THeorem}%

\begin{document}

\title[On Some Closure Properties of \textit{nc-eNCE} Graph Grammars]{On Some Closure Properties of \textit{nc-eNCE} Graph Grammars}


\author[1,2]{\fnm{Jayakrishna} \sur{Vijayakumar}}\email{vjayakrishna@amaljyothi.ac.in}
\author*[3]{\fnm{Lisa} \sur{Mathew}}\email{lisamathew@amaljyothi.ac.in}

\affil[1]{\orgdiv{Department of Computer Science and Engineering}, \orgname{Amal Jyothi College of Engineering}, \orgaddress{\street{Kanjirappally}, \city{Kottayam}, \postcode{686 518}, \state{Kerala}, \country{India}}}

\affil[2]{\orgdiv{Research scholar}, \orgname{APJ Abdul Kalam Technological University}, \orgaddress{\street{CET Campus}, \city{Thiruvananthapuram}, \postcode{695 016}, \state{Kerala}, \country{India}}}

\affil[3]{\orgdiv{Department of Basic Sciences }, \orgname{Amal Jyothi College of Engineering}, \orgaddress{\street{Kanjirappally}, \city{Kottayam}, \postcode{686 518}, \state{Kerala}, \country{India}}}

\abstract{In the study of automata and grammars, closure properties of the associated languages have been studied extensively. In particular, closure properties of various types of graph grammars have been examined in (Rozenberg and Welzl, Inf. and Control,1986) and (Rozenberg and Welzl, Acta Informatica,1986). In this paper we examine some critical closure properties  of the  \textit{nc-eNCE}  graph grammars discussed in (Jayakrishna and Mathew, Symmetry 2023) and (Jayakrishna and Mathew, ICMICDS 2022).
}


\maketitle
\section{Introduction}
Graph grammars  have been a challenging area of research  in theoretical computer science  since their introduction in the sixties \cite{Ehrig_1978,Ehrig_1986,Ehrig_1992,Engelfriet05,Kreowski_1986,pavlidis1966,Rozenberg_97}. They provide   a framework for specifying the structure and transformation of graphs and hence are capable of generating a wide range of graphs using a small set of rules.  Thus the creation of complex structures that would be difficult to specify using other methods becomes possible. They  are often used in computer science to specify the structural and behavioral aspects of software systems. For instance, they can be used to specify the structure of a database schema \cite{Ehrig_1978}, a programming language or a software architecture. They also find application  in other fields, such as biology, to model the structure and behavior of biological systems. Moreover, they can be used to specify the behavior of a system over time,  thus facilitating   the modeling of dynamic systems.
\par  Graph grammars consist primarily of a set of rules which specify how to iteratively transform a graph by adding, deleting, or modifying its nodes and edges- \cite{Rozenberg_97,Roz_1997}. A typical graph grammar consists of $(i)$  a start graph $S$ and  $(ii)$ a set $P$ of transformation rules (productions).  These rules  help to convert the host graph $H$ (initially the host graph is $S$) to the desired transformed graph. The rules  have an crucial role in the transformation and are of the form $P(M, D, E)$, where $M$ is the mother graph to be replaced (present in $H$), $D$ (daughter graph) is the graph  which replaces $M$ , 
and $E$ is the embedding mechanism which describes how the daughter graph $D$ is connected to the the remaining portion of the host graph $(H'=H-M)$ . 

\par  On the other hand, a class of graph languages can have two types of operations defined on it- set operations and operations on graphs belonging to the language. Set operations include union, intersection and complementation while operations on graphs include disjoint sum (graph theoretical union), graph complementation and others including some new operations which we define here.  Closure properties of some special types types of graph grammars have been examined in \cite{boundary,closure}  Here, we examine some salient closure properties  of the  \textit{nc-eNCE}  graph grammars discussed in \cite{Jayakrishna_20}and \cite{jk_jumping}.
\par In Section \ref{Prelim}  we recall some preliminary notions about \textit{nc-eNCE}.  Section \ref{op}  of this paper  introduces a variety of operations on Graph Languages and a few other related concepts. In Section \ref{clos} we prove that these languages are closed under the operations of union, disjoint sum and  Kleene sum. Closure properties under two more operations - Chain concatenation and star concatenation are also discussed. Since the   proofs of the closure properties with respect to these two operations employ  jumping connection instructions the class of jumping \textit{nc-eNCE} graph grammars are closed under these operations. We conclude our discussion in Section \ref{conclo} 
\section{Preliminaries }\label{Prelim}
Gluing and connecting are the two popular mechanisms used to embed the daughter graph in the mother graph.In the gluing approach \cite{Ehrig_1978,Ehrig_1986,Ehrig_1992,Kreowski_1986}, certain nodes and edges  of $H'$ and of $D$ are identified for fusion. The edges thus identified are those which were  originally incident to $M$ in $H$. There are different ways of identifying the corresponding nodes in the host and daughter graphs. One way is to use labels or attributes attached to the nodes/edges. Another way is to use a mapping function that associates each node/edge in the daughter graph with a node in the host graph. 
\par  Hierarchical gluing is a technique that allows us to glue multiple graphs in a hierarchical manner. In this approach, we first glue two or more graphs to form a larger graph and then repeat the process with the remaining graphs until the desired graph is obtained. The gluing approach is commonly used in graph compression and graph visualization, where it can simplify the display of complex graphs by representing subgraphs with a single node. 

The connecting approach is another way of embedding  daughter graphs in the host graph. In this approach, the daughter graph is connected to the host graph by adding edges between designated nodes in the two graphs. The edges represent the relationships between the nodes in the daughter graph and the nodes in the host graph.

The connecting approach is suitable when the embedded graph is small compared to the host graph. In this case, it may not be necessary to replace nodes in the host graph with the daughter graph. Instead, the relationships between the nodes in the daughter graph and the nodes in the host graph can be represented using edges. This  is known as  edge replacement. The connecting approach is usually used in graph query processing, where it is used to match smaller graphs against a larger graph. It may also  be used in graph editing tools, where it can be used to insert and delete subgraphs within a larger graph.

 The gluing approach is advantageous when the embedded graph needs to be modified frequently or when the embedded graph is large compared to the host graph. The gluing approach is generally more expensive than the connecting approach   in terms of computational complexity, since  it requires the replacement of nodes in the host graph with the daughter graph. The connecting approach, on the other hand, only requires the addition of edges between the daughter graph and the host graph.

\subsection{Graph Grammars with Regular Control}
Graph grammars with regular control are a variant of graph grammars that use regular expressions to specify the control flow of the transformation rules. Regular expressions are used to specify the conditions under which a production rule can be applied. The evolution of graphs are determined by the regular expression which controls the  order of application of  rules. Such graph grammars are known as non-confluent graph grammars. One such variant of this type specified in \cite{Jayakrishna_20} is Non-confluent Edge and Node Controlled Embedding $(nc-eNCE)$ Graph Grammar.

Non-confluent Edge and Node Controlled Embedding $(nc-eNCE)$ Graph Grammar \cite{Jayakrishna_20}: Restricting the sequence of production rules utilized results in this type of graph grammar that introduces a small amount of determinism to the essentially non-deterministic concept of a graph grammar. 

\begin{definition} \cite{jkicmc21, Jayakrishna_20}
\label{ncence}
A construct $ncG= $( $\Sigma, \Delta, \Gamma, \Omega, P, G_S, R(P) $) is known as an $nc-eNCE$ graph grammar where
\begin{itemize}
\item []$\Sigma$  and $\Gamma$  are  sets of symbols used to label nodes and edges respectively,
\item []$\Delta $ and $\Omega$ are  the collections of terminal symbols in $\Sigma$ and $\Gamma$ respectively, 
\item []A production rule in $P$, $p:  A \rightarrow ( D, C )$ acting on the mother node $M$ with label $A$ has a collection C of connection instructions  $(a,p\mid q ,B)$ associated with it. Here $x$ with label $a$ is a neighbor of $M$ and $B$ is a node in $D$. The edge $p$ which connected $x$ and $M$ is removed and  a new edge $q$ is established between $x$ and $B$.
\item []$G_S$ is the initial graph,
\item []The regular control, $R(P)$, regulates the sequence of application of the production rules.
\end{itemize}
\end{definition}
Several variants of $nc-eNCE$ graph grammars such as nc-$eNCE$ graph grammar with deletion $(dnc-eNCE)$,  nc-$eNCE$ graph grammar with $\psi$ labelled edges $(\psi nc-eNCE)$, etc have been discussed in \cite{Jayakrishna_20}.
\subsection{Jumping Graph Grammar}
A jumping graph grammar ($JGG$) differs from its conventional counterpart(jumping gramars) \cite{kvrivka2015jumping, medunazemek2014} in the fact that in addition to the normal connections established as part of daughter graph embedding with vertices which were adjacent to the vertices of the mother graph,  we can also create extra edges known as jumping edges to establish connections from vertices of daughter graph to other vertices  in the rest of the host graph. The formal definition of jumping graph grammars is as follows.\\
\begin{definition}\label{Generaljgg}\cite{jk_jumping}
A Jumping Graph Grammar is a  construct $JGG=$($\Sigma, P, C, S$), where
\begin{itemize}
\item [] $\Sigma$ is the finite set of node labels
\item []  $P$ is the finite set of graph production rules
\item []   $C$ is the connection instruction
\item []  $G_S$ is the start/initial graph
\end{itemize}
Here $P$ is of the form $P: M\rightarrow D$ where $M$ is the mother graph to  be replaced by daughter graph $D$.

\end{definition}
In general graph grammars, the connection instruction $C$ establishes edges from the designated nodes of $D$ to the nodes of $M$ which were having incident edges to the mother graph that were removed during replacement. But in jumping graph grammars, the connection instruction has an extended capability to establish edges from designated nodes of $D$ to any of the nodes present in $G_S \backslash M$.\\
The language represented by the jumping graph grammar $JGG$ is \\
$L$($JGG$)$~={ \{G | G_S \overset{*}{\implies} G \} }$ where, $G$ is the set of graphs obtained from $G_S$ by recursively applying rules in $P$.

\subsection{Non-confluent Edge and Node Controlled Embedding Jumping Graph grammars (\textit{nc-eNCE-JGG})}
We extend the definition of $nc-eNCE$ graph grammar in \cite{jkicmc21} for defining this new grammar.
 Formally, we have:
\begin{definition} \label{Def2}
An $nc-eNCE-JGG$ graph grammar is a 7 tuple: $ncJGG= \text{(}\Sigma, \Delta, \Gamma, \Omega, $ $P, G_S, R\text{(}P\text{)} $\text{)} where
\begin{itemize}
\item []$\Sigma$ is the set of node labels,
\item []$\Delta$ is the set of terminal node labels,
\item []$\Gamma$ is the set of edge labels,
\item []$\Omega$ is the set of terminal  edge labels,
\item []The productions in $P$ are of  one of the  forms $p:  A \rightarrow \text{(} D, C $\text{)} where $A$ is the label of the mother node and $D$ is the daughter graph.
The connection instruction $C$ can be in any one of the following forms
\begin{enumerate}
\item ($a,p\mid q ,B$), where $p$ and $q$ are edge labels, $a$ is the node label of one of the neighbours of the mother node and $B$ is a node in $D$. The interpretation is that we find an edge labelled $p$ in the host graph which had connected a node $x$ labelled $a$ to the mother node and create a new edge labelled $q$ between $x$ and the node $B$ in the daughter graph
    \item  ($a,\alpha,b$\text{)}, where $a$ is the node label of one of the nodes in $D$, $b$ is the label of any of the nodes in remaining graph after removing the mother node and $\alpha$ is the label of the new edge connecting $a$ and $b$.    
\end{enumerate}
\item []$G_S$ is the start/initial graph,
\item []$R\text{(}P$) is a regular control which specifies the order of application of the productions,
Hence this grammar becomes restricted or non-confluent.
\end{itemize}
\end{definition}
The language represented by this grammar is $L\text{(}ncJGG\text{)}={ \{ G \in G_\Delta |G_S \overset{R\text{(}P\text{)}}{\implies}G \} }$ where {$G_\Delta$} is the set of graphs containing only terminal nodes. Hence the language of an $nc-eNCE-JGG$ graph grammar is a  set of graphs whose nodes have terminal labels, generated by applying a series of productions in the order $p_1, p_2, \cdots, p_n$ where $p_1p_2\cdots p_n$ is a word in the language represented by the regular control $R\text{(}P$).
\section{Operations on Graph Languages}\label{op}
\begin{definition}
 The union of two  graph languages $L_1$ and $L_2$ is the graph language obtained by taking the usual union of the two families of graphs i.e. $L=L_1\cup L_2$
\end{definition}
\begin{definition}
 The disjoint sum of two  graph languages $L_1$ and $L_2$ is the graph language$L_1\oplus L_2$ obtained by taking the disjoint sum of sample graphs from  two languages i.e. $L_1\oplus L_2=\{G_1+ G_2|G_1\in L_1,G_2\in L_2\}$.
\end{definition}
\begin{definition}
 The Kleene sum of a  graph language $L$ is the graph language obtained by taking the disjoint sum of  a countable number of sample graphs from $L$  i.e. $L^{\oplus}=\{G_1+ G_2+\cdots G_n|G_i\in L,1\leq i\leq n\}$
\end{definition}
\begin{definition}
    Let $L$ be a graph language generated by a graph grammar $G$, the set of graphs obtained by ignoring the labels of graphs in $L$ forms the  underlying graph language $U(L)$.
\end{definition}
\begin{definition}
 The chain concatenation  of two   graph languages $L_1$ and $L_2$ which do not contain null graphs is the graph language obtained by  joining the two designated nodes chosen from sample graphs $G_1$ and $G_2$ in $L_1$ and $L_2$ respectively and is denoted  by$L_1\odot L_2$
 \end{definition}
 \begin{definition}
 The star concatenation of a  graph language $L_1$ is the graph language obtained by connecting  designated nodes  in  a countable number of sample graphs from $L_1$ to a central node. This is denoted by $L_1^ \circledast$.  
\end{definition}
\section {Closure Properties of Graph Languages Generated Using \textit{nc-eNCE} Graph Grammar}\label{clos}
The class of graph languages generated using $nc-eNCE$ Graph Grammar is closed under union,  disjoint sum, Kleene sum and chain concatenation.
\begin{theorem}
The class of string languages generated using nc-eNCE graph grammars is closed under the operation union.
\end{theorem}
\begin{proof}
 
 Consider two $nc-eNCE$ graph grammars $ ncG_1 = ( \Sigma_1, \Delta_1, \Gamma_1, \Omega_1, P_1, S_1,$ $  R_1(P_1))$  and $ncG_2 = ( \Sigma_2, \Delta_2, \Gamma_2, \Omega_2, P_2, S_2, R_2(P_2))$   generating graph languages $L_1$ and $L_2$ respectively. Assuming that $S$ is a new symbol that does not appear in $\Sigma_1 \cup \Sigma_2$, we define a new $nc-eNCE$ graph grammar $ncG$ capable of generating the language $L = L_1 U L_2$ as follows.
 \vspace{-.3cm}

\[ncG=( \Sigma, \Delta, \Gamma, \Omega, P, S, R(P) )~~ \text{where}\],
\vspace{-1cm}
\begin{itemize}
\item []$\Sigma = \Sigma_1 \cup \Sigma_2\cup\{S\}$
\item[] $\Delta = \Delta_1 \cup \Delta_2$
\item[] $\Gamma = \Gamma_1 \cup \Gamma_2$
\item[] $\Omega =\Omega_1 \cup \Omega_2$
\item[] $P        = P_1 \cup P_2 \cup \{p_{01}, p_{02}\}$
\item[] R(P)= $[p_{01}.R_1(P_1)]+[p_{02}.R_2(P_2)]$
\end{itemize}
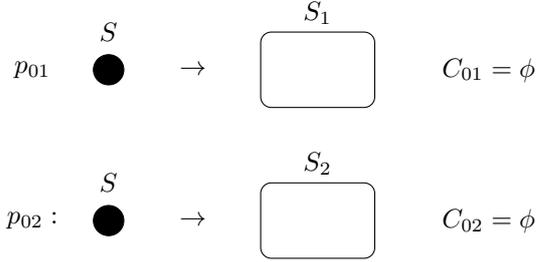
\begin{figure}[H]
    \centering
    \begin{tikzpicture}
     
\node at (0,2.5){$p_{01}$};
\draw [fill=black] (1,2.5)  circle (.2cm);
\node at (1,3){$S$};
\node at (2,2.5) {$~\rightarrow$};
\draw[rounded corners] (3, 2) rectangle (4.5, 3);   
\node at (3.75,3.25){$S_1$};
\node at (6,2.5) {$C_{01}= \phi $};
\node at (0,.5){$p_{02}:$};
\draw [fill=black] (1,.5)
circle (.2cm);
\node at (1,1){$S$};
\node at (2,.5) {$~\rightarrow$};
\draw[rounded corners] (3, 0) rectangle (4.5, 1); 
\node at (3.75,1.25){$S_2$};
\node at (6,.5) {$C_{02}= \phi $};
 
\end{tikzpicture}
\caption{New Production for generation of union of graph languages}
\label{f13}
\end{figure}

Here $S_1$ and $S_2$ are the start graphs of $ncG_1$ and $ncG_2$ respectively. Clearly any graph in $L_1$ can be generated using $[p_{01}.R_1(P_1)]$ while any graph in $L_2$ can be generated using $[p_{02}.R_2(P_2)]$.
\end{proof}

\begin{theorem}
The class of graph languages generated using $nc-eNCE$ graph grammars is closed under the graph operation disjoint sum.
\end{theorem}
\begin{proof}
Consider two  $nc-eNCE$ graph grammars $ncG_1 = (\Sigma_1, \Delta_1, \Gamma_1, \Omega_1, P_1, S_1, $ $ R_1(P_1) )$ and $ncG_2 = ( \Sigma_2, \Delta_2, \Gamma_2, \Omega_2, P_2, S_2, R_2(P_2) )$  generating graph languages $L_1$ and $L_2$ respectively. Without loss of generality we can assume that $P_1$ and $P_2$ have no elements in common.We define a new $nc-eNCE$ graph grammar $ncG$ capable of generating the language $L = L_1 \oplus L_2$ as follows.
\[ncG=( \Sigma, \Delta, \Gamma, \Omega, P, S, R(P) )~~ \text{where}\],
\vspace{-1cm}
\begin{itemize}
\item[] $\Sigma = \Sigma_1 \cup \Sigma_2\cup\{S\}$
\item[] $\Delta = \Delta_1 \cup \Delta_2$
\item[] $\Gamma = \Gamma_1 \cup \Gamma_2$
\item[] $\Omega =\Omega_1 \cup \Omega_2$
\item[] $P        = P_1 \cup P_2 \cup \{p_{01}\}$
\item[] $R(P)= p_{01}.R_1(P_1).R_2(P_2)$
\end{itemize}
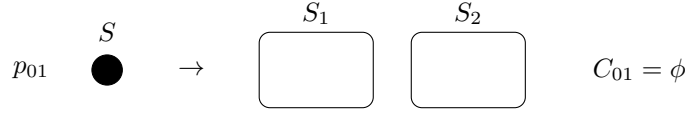
\begin{figure}[H]
    \centering
\begin{tikzpicture}
     
\node at (0,2.5){$p_{01}$};
\draw [fill=black] (1,2.5)  circle (.2cm);
\node at (1,3){$S$};
\node at (2,2.5) {$~\rightarrow$};
\draw[rounded corners] (3, 2) rectangle (4.5, 3);   
\node at (3.75,3.25){$S_1$};
\draw[rounded corners] (5, 2) rectangle (6.5,3);   
\node at (5.75,3.25){$S_2$};
\node at (8,2.5) {$C_{01}= \phi $};
\end{tikzpicture}
    \caption{Additional Production for generation of disjoint Sum}  
\label{pro_dis_sum}
\end{figure}

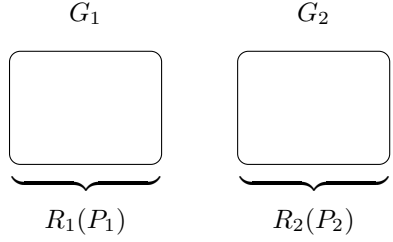
\begin{figure}[H]
    \centering
\begin{tikzpicture}
     \draw[rounded corners] (0, 0) rectangle (2, 1.5);  
     \draw[rounded corners] (3, 0) rectangle (5, 1.5);
\node at (1.,2){$G_1$};
\node at (4,2){$G_2$};
\node at (1.,-.25){$\underbrace{\kern5.5em } $};
\node at (4,-.25){$\underbrace{\kern5.5em } $};
\node at (1.,-.75){$R_1(P_1)$};
\node at (4,-.75){$R_2(P_2)$};
    \end{tikzpicture}
    \caption{ Schemata of Disjoint sum with graphs $G_1 \& G_2$ from $L_1$ and $L_2$ respectively.}
    \label{scheme_disjsum}
\end{figure}

The start graphs $S_1$ and $S_2$ of $ncG_1$ and $ncG_2$ are initially generated using $p_{01}$. This is followed by $R_1(P_1)$ generating $G_1 \in L_1$ and $R_2(P_2)$ generating $G_2 \in L_2$. Hence it is clear that $R(P)$ generates $L_1 \oplus L_2$.
\end{proof}
\begin{theorem}
The class of unlabelled graph languages generated using $nc-eNCE$ graph grammars is closed under the graph operation Kleene sum.
\end{theorem}
\begin{proof}
Consider an  $nc-eNCE$ graph grammar $ncG_1 = (\Sigma_1, \Delta_1, \Gamma_1, \Omega_1, P_1, S_1, $ $ R_1(P_1) )$   generating the graph language $L_1$. We define a new $nc-eNCE$ graph grammar $ncG$ capable of generating the language $L= L_1^\oplus$ as follows.
\[ncG=( \Sigma, \Delta, \Gamma, \Omega, P, S, R(P) )~~ \text{where}\],
\vspace{-1cm}
\begin{itemize}
\item[] $\Sigma = \Sigma_1 \cup\{S\}$
\item[] $\Delta = \Delta_1$
\item[] $\Gamma = \Gamma_1$
\item[] $\Omega =\Omega_1$
\item[] $P        = P_1 \cup \{p_{01},p_{02}\}$
\item[] $R(P)= (p_{01}.R_1(P_1))^*.p_{02}$
\end{itemize}

\begin{figure}[H]
    \centering

\begin{tikzpicture}
     
\node at (0,2.5){$p_{01}$};
\draw [fill=black] (1,2.5)  circle (.2cm);
\node at (1,3){$S$};
\node at (2,2.5) {$~\rightarrow$};
\draw[rounded corners] (3, 2) rectangle (4.5, 3);   
\node at (3.75,3.25){$S_1$};
\draw [fill=black] (6,2.5)  circle (.2cm);
\node at (6,3){$S$};
\node at (8,2.5) {$C_{01}= \phi $};
\node at (0,.5){$p_{02}:$};
\draw [fill=black] (1,.5)
circle (.2cm);
\node at (1,1){$S$};
\node at (2,.5) {$~\rightarrow$};
\node at (3.5, 0.5){$\varepsilon$}; 
\node at (8,.5) {$C_{02}= \phi $};
\end{tikzpicture}\caption{Additional productions for generation of Kleene sum}  
\label{pro_kleen_sum}
\end{figure}
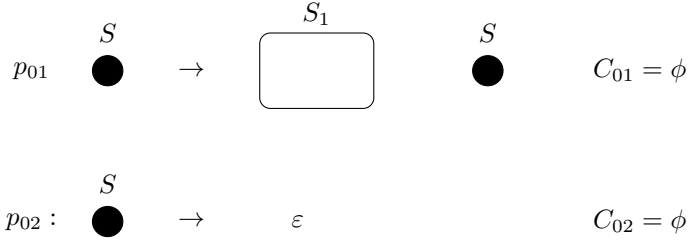

\begin{figure}[H]
    \centering
     \begin{tikzpicture}
     \draw[rounded corners] (0, 0) rectangle (1.5, 1.);    
     \draw[rounded corners] (3, 0) rectangle (5, 1.);
     \draw[rounded corners] (8, 0) rectangle (10,1);
\node at (1,1.75){$G_1$};
\node at (4,1.75){$G_2$};
\node at (6.5,.75){$\cdots$};
\node at (9,1.75){$G_n$};
    \end{tikzpicture}

    \caption{ Schemata of Kleene sum with graphs $G_1, G_2 \cdots G_n$ from $L_1$.}
    \label{Schem_kleensum}
\end{figure}
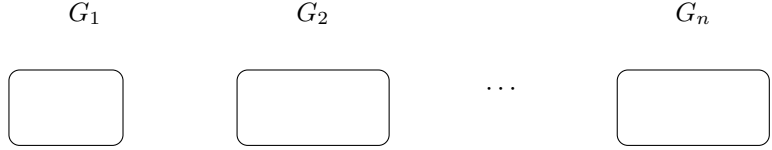

\par Initially the production $p_{01}$ replaces the start node $S$ of $ncG$ with a copy of the start graph $S_1$ of $ncG_1$ and the node $S$ itself. Now $R_1(P_1)$ replaces $S_1$ by a graph $G_1 \in L_1$. The rules $p_{01}$ and $R_1(P_1)$ are recursively applied to further generate new components of the resultant graph. Finally the derivation terminates when $p_{02}$ is applied. Hence it is clear that $R(P)$ generates $L_1^ \oplus$.
\end{proof}

\begin{theorem}
The class of underlying graph languages formed from graph classes generated   using $nc-eNCE$ jumping graph grammars is closed under the graph operation  chain concatenation.\end{theorem}
\begin{proof}
 
 Consider two  $nc-eNCE$ graph grammars $ncG_1 = (\Sigma_1, \Delta_1, \Gamma_1, \Omega_1, P_1, S_1, $ $ R_1(P_1) )$ and $ncG_2 = ( \Sigma_2, \Delta_2, \Gamma_2, \Omega_2, P_2, S_2, R_2(P_2) )$  generating graph languages $L_1$ and $L_2$ respectively. Without loss of generality we can assume that $P_1$ and $P_2$ have no elements in common. We begin our construction of the required grammar by  carrying out the following steps.
 \begin{enumerate}
     \item Replace each $p_j^*$ in the regular expressions $ R_i(P_i)$ by the equivalent expression $(p_j^*p_j + \varepsilon)$   and simplifying so that $ R_i(P_i)$ is now of the form $R_{i1}(P_i) +R_{i2}(P_i) +\cdots R_{in_i}(P_i)$. This ensures that each $R_{ik}(P_i)$ does not contain a the operator '$+$'.
     \item \label{new} Initialize ${P'}_1= P_1$.  For each $R_{1k}(P_1) ,1\leq k\leq n_1$,  do the following
     \begin{enumerate}
         \item Since every graph generated by  $R_{1k}(P_1)$ is  non null  it contains atleast one vertex with a terminal label. Identify the last production rule $p_{1l}$ in the series  $R_{1k}(P_1)$ which has atleast one node labelled with a terminal symbol, say '$a$' on the right hand side. 
         \item Create a copy of $p_{1l}$  in which one occurrence of '$a$' is replaced by a new terminal symbol '$x$', $x \notin (\Delta_1 \cup \Delta_2)$ and label this production ${p'}_{1l}$. The connection instructions corresponding to ${p'}_{1l}$ includes all the instructions in ${p}_{1l}$. In addition every connection instruction which contains '$a$' is duplicated and '$x$' is inserted instead of '$a$'. Update ${P'}_1$ as ${P'}_1 = {P'}_1\cup {p'}_{1l}$. 
                        \end{enumerate}
 \item   Repeat Step \ref{new} for $R_{2k}(P_2),1\leq k\leq n_2$ with the terminal symbol '$b$' on the right hand side of $p_{1l}$ replaced by $y$, $y \notin (\Delta_1 \cup \Delta_2)$ label this production ${p'}_{1l}$. The connection instructions corresponding to ${p'}_{1l}$ includes all the instructions in ${p}_{1l}$. In addition every connection instruction which contains '$b$' is duplicated and '$y$' is inserted instead of '$b$'.One more jumping connection instruction of the form $(y,\alpha,x)$ is added to this set. Update ${P'}_2$ as ${P'}_2 = {P'}_2\cup {p'}_{1l}$ .
 \item  Replace the last occurrence  of $p_{il}$ with ${p'}_{il}$ in each $R_{ik}(P_i)$ to obtain ${R'}_{ik}(P_i)$. The new regular control is given by 
 \[{R'}_i({P'}_i)={R'}_{i1}({P'}_i) +{R'}_{i2}({P'}_i) +\cdots {R'}_{in_i}({P'}_i)\]
 \end{enumerate}.
 
We  now define a new $nc-eNCE$ graph grammar $ncG$ capable of generating the language $L$ such that $ U(L)=U(L_1 \odot L_2)$ as follows.
\[ncG=( \Sigma, \Delta, \Gamma, \Omega, P, S, R(P) )~~ \text{where}\],
\vspace{-1cm}
\begin{itemize}
\item[] $\Sigma = \Sigma_1 \cup \Sigma_2 \cup \{S,x,y\}$
\item[] $\Delta = \Delta_1 \cup \Delta_2 \cup \{x,y\}$
\item[] $\Gamma = \Gamma_1 \cup \Gamma_2$
\item[] $\Omega =\Omega_1 \cup \Omega_2$
\item[] $P        = {P'}_1 \cup {P'}_2 \cup \{p_{01}\}$
\item[] $R(P)=p_{01}{R'}_1({P'}_1){R'}_2({P'}_2)$
\end{itemize}
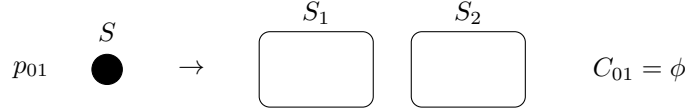
\begin{figure}[H]
    \centering
    \begin{tikzpicture}
     
\node at (0,2.5){$p_{01}$};
\draw [fill=black] (1,2.5)  circle (.2cm);
\node at (1,3){$S$};
\node at (2,2.5) {$~\rightarrow$};
\draw[rounded corners] (3, 2) rectangle (4.5, 3);   
\node at (3.75,3.25){$S_1$};
\draw[rounded corners] (5, 2) rectangle (6.5,3);   
\node at (5.75,3.25){$S_2$};
\node at (8,2.5) {$C_{01}= \phi $};
\end{tikzpicture}
\caption{New Production for generation of base graph for chain concatenation.}  
\label{f16}
\end{figure}

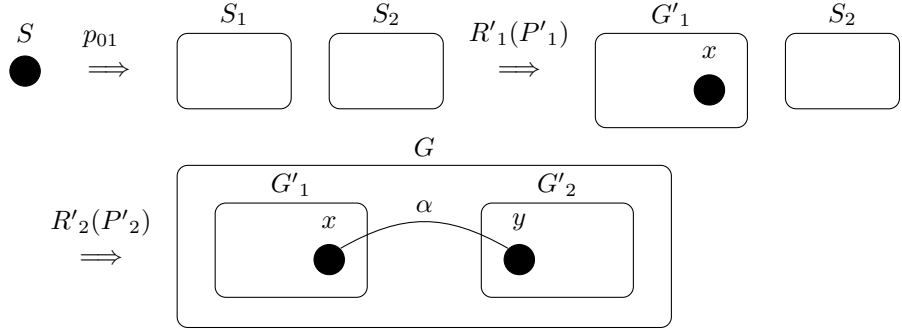
\begin{figure}[H]
    \centering
 \begin{tikzpicture}
     
\draw [fill=black] (1,2.5)  circle (.2cm);
\node at (1,3){$S$};
\node at (2,2.5) [label=above:$p_{01}$]{$~\Longrightarrow$};
\draw[rounded corners] (3, 2) rectangle (4.5, 3);   
\node at (3.75,3.25){$S_1$};
\draw[rounded corners] (5, 2) rectangle (6.5,3);   
\node at (5.75,3.25){$S_2$};
\node at (7.5,2.5)[label=above:${R'}_1({P'}_1)$] {$\Longrightarrow$};
\draw[rounded corners] (8.5, 1.75) rectangle (10.5, 3);
\node at (10,2.75){$x$};
\draw (10,2.25)[fill=black] circle  (.2 cm);
\node at (9.5,3.25){${G'}_1$};
\draw[rounded corners] (11, 2) rectangle (12.5, 3);
\node at (11.75,3.25){$S_2$};
\node at (2,0) [label=above:${R'}_2({P'}_2)$] {$\Longrightarrow$};
\draw[rounded corners] (3., 1.25) rectangle (9.5, -.9);
\node at (6.25,1.5){$G$};
\draw[rounded corners] (3.5, .75) rectangle (5.5, -.5);
\node at (4.5,1){${G'}_1$};
\node at (5,.5){$x$};
\draw (5,0) [fill=black]circle  (.2 cm);
\draw[rounded corners] (7, .75) rectangle (9, -.5);
\node at (8,1){${G'}_2$};
\node at (7.5,.5){$y$};
\draw (7.5,0) [fill=black]circle  (.2 cm);
\path((5.15,.15) edge [bend left, looseness=1.1] node[above,midway]{$\alpha$} (7.35,.15);
\end{tikzpicture}
    \caption{Schemata of chain concatenation of graphs from $L_1$ and $L_2 $using $ncG$}
    \label{f17}
\end{figure}
It is clear that the underlying graph obtained after  chain concatenation of graphs $G_1$  and $G_2$  is  identical to the the underlying graph $G$  which is  the result of this construction.
\end{proof}
\begin{theorem}
The class of underlying graph languages formed from graph classes generated   using $nc-eNCE$ jumping graph grammars is closed under the graph operation  star concatenation.\end{theorem}
\begin{proof}
Consider an  $nc-eNCE$ graph grammar $ncG_1 = (\Sigma_1, \Delta_1, \Gamma_1, \Omega_1, P_1, S_1, $ $ R_1(P_1) )$   generating the graph language $L_1$. We define a new $nc-eNCE$ graph grammar $ncG$ capable of generating the language $L= L_1^\circledast$ as follows.
\[ncG=( \Sigma, \Delta, \Gamma, \Omega, P, S, R(P) )~~ \text{where}\],
\vspace{-1cm}
\begin{itemize}
\item[] $\Sigma = \Sigma_1 \cup\{S,x,c\}$
\item[] $\Delta = \Delta_1 \cup\{x,c\}$
\item[] $\Gamma = \Gamma_1$
\item[] $\Omega =\Omega_1$
\item[] $P = {P'}_1 \cup \{p_{01},p_{02}\}$
\item[] $R(P)= (p_{01}.{R'}_1({P'}_1))^*.p_{02}$
\end{itemize}

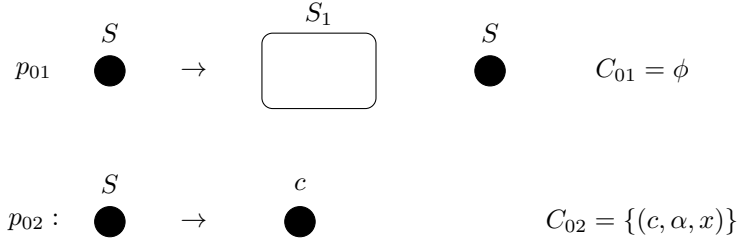
\begin{figure}[H]
    \centering
  \begin{tikzpicture}
     
\node at (0,2.5){$p_{01}$};
\draw [fill=black] (1,2.5)  circle (.2cm);
\node at (1,3){$S$};
\node at (2,2.5) {$~\rightarrow$};
\draw[rounded corners] (3, 2) rectangle (4.5, 3);   
\node at (3.75,3.25){$S_1$};
\draw [fill=black] (6,2.5)  circle (.2cm);
\node at (6,3){$S$};
\node at (8,2.5) {$C_{01}= \phi $};
\node at (0,.5){$p_{02}:$};
\draw [fill=black] (1,.5)
circle (.2cm);
\node at (1,1){$S$};
\node at (2,.5) {$~\rightarrow$};
\draw [fill=black] (3.5, 0.5)
circle (.2cm);
\node at (3.5, 1){$c$}; 
\node at (8,.5) {$C_{02}= \{(c,\alpha,x)\} $};
\end{tikzpicture}
\caption{Additional Productions for generation of star concatenation}  
\label{pro_star_cat}
\end{figure}
Initially the production $p_{01}$ replaces the start node $S$ of $ncG$ with a copy of the start graph $S_1$ of $ncG_1$ and the node $S$ itself. Now ${R'}_1({P'}_1)$ replaces $S_1$ by a graph $G_1 \in L_1$. The rules $p_{01}$ and ${R'}_1({P'}_1)$ are recursively applied to further generate new components of the resultant graph. Finally the derivation terminates when $p_{02}$ is applied. Hence it is clear that $R(P)$ generates $L_1^ \circledast$.
\end{proof}
\begin{figure}[H]
    \centering
\begin{tikzpicture}
\draw [fill=black] (1,2.5)  circle (.2cm);
\node at (1,3){$S$};
\node at (1.7,2.5) [label=above:$p_{01}$]{$~\Longrightarrow$};
\draw[rounded corners] (2.5, 2) rectangle (4, 3);   
\node at (3.25,3.25){$S_1$};
\draw [fill=black] (4.5,2.5)  circle (.2cm);
\node at (4.7,3){$S$};  
\node at (5.75,2.5)[label=above:${R'}_1({P'}_1)$]  {$~\Longrightarrow$};
\draw[rounded corners] (6.5, 2) rectangle (8, 3.25);   
\draw [fill=black] (9,2.5)  circle (.2cm);
\node at (7.5,3.5){${G'}_1$};
\node at (7.5,3){$x$};
\draw (7.5,2.5) [fill=black]circle  (.2 cm);

\node at (9,3){$S$};  
\draw[rounded corners] (2.5, .75) rectangle (4.5, -.5);
\node at (3.5,1){${G'}_1$};
\node at (3.5,.5){$x$};
\draw (3.5,0) [fill=black]circle  (.2 cm);
\draw[rounded corners] (5, .75) rectangle (7, -.5);
\node at (6,1){${G'}_2$};
\node at (6,.5){$x$};
\draw (6,0) [fill=black]circle  (.2 cm);
\draw[rounded corners] (7.5, .75) rectangle (9.5, -.5);
\node at (8.5,1){${G'}_3$};
\node at (8.5,.5){$x$};
\draw (8.5,0) [fill=black]circle  (.2 cm);
\node at (.5,-1.5) [label=above:$(p_{01}.{R'}_1({P'}_1))^{n-1}$] {$\Longrightarrow$};
\node at (2.15,-1.5){${G'}_4$};
\draw[rounded corners] (2.5, -.75) rectangle (4.5, -2);
\node at (3.5,1){${G'}_1$};
\node at (3.5,-1){$x$};
\draw (3.5,-1.5) [fill=black]circle  (.2 cm);
\node at (6,-1){$S$};
\draw[rounded corners] (2.5, -2.25) rectangle (4.5, -3.5);
\node at (3.5,-2.5){$x$};
\node at (3.5,-3.75){${G'}_6$};
\draw (3.5,-3) [fill=black]circle  (.2 cm);
\node at (6,-3){$\cdots$};
\draw[rounded corners] (7.5, -2.25) rectangle (9.5, -3.5);
\node at (8.5,-2.5){$x$};
\draw (8.5,-3) [fill=black]circle  (.2 cm);
\node at (8.5,-3.75){${G'}_n$};
\draw[rounded corners] (7.5, -.75) rectangle (9.5, -2);
\node at (8.5,-1){$x$};
\draw (8.5,-1.5) [fill=black]circle  (.2 cm);
\node at (9.85,-1.5){${G'}_5$};
\draw (6,-1.5) [fill=black]circle  (.2 cm);
\end{tikzpicture}
\begin{tikzpicture}
  \node at (6.25,1.5){$G$};
\draw[rounded corners] (1.75, 1.25) rectangle (10.25, -4);
\draw[rounded corners] (2.5, .75) rectangle (4.5, -.5);
\node at (3.5,1){${G'}_1$};
\node at (3.5,.5){$x$};
\draw (3.5,0) [fill=black]circle  (.2 cm);
\draw[rounded corners] (5, .75) rectangle (7, -.5);
\node at (6,1){${G'}_2$};
\node at (6,.5){$x$};
\draw (6,0) [fill=black]circle  (.2 cm);
\draw[rounded corners] (7.5, .75) rectangle (9.5, -.5);
\node at (8.5,1){${G'}_3$};
\node at (8.5,.5){$x$};
\draw (8.5,0) [fill=black]circle  (.2 cm);
\node at (-.25,-1.5) [label=above:$p_{02}$] {$\Longrightarrow$};
\node at (2.15,-1.5){${G'}_4$};
\draw[rounded corners] (2.5, -.75) rectangle (4.5, -2);
\node at (3.5,1){${G'}_1$};
\node at (3.5,-1){$x$};
\draw (3.5,-1.5) [fill=black]circle  (.2 cm);
\draw[rounded corners] (2.5, -2.25) rectangle (4.5, -3.5);
\node at (3.5,-2.5){$x$};
\node at (3.5,-3.75){${G'}_6$};
\draw (3.5,-3) [fill=black]circle  (.2 cm);
\node at (6,-3){$\cdots$};
\node at (6,-2){$c$};
\draw[rounded corners] (7.5, -2.25) rectangle (9.5, -3.5);
\node at (8.5,-2.5){$x$};
\draw (8.5,-3) [fill=black]circle  (.2 cm);
\node at (8.5,-3.75){${G'}_n$};
\draw[rounded corners] (7.5, -.75) rectangle (9.5, -2);
\node at (8.5,-1){$x$};
\draw (8.5,-1.5) [fill=black]circle  (.2 cm);
\node at (9.85,-1.5){${G'}_5$};
\draw (6,-1.5) [fill=black]circle  (.2 cm); 
\draw[thick] (3.7,-1.5)  --(5.8,-1.5);
\draw[thick] (6.2,-1.5)  --(8.3,-1.5);
\draw[thick] (6,-1.5)  --(6,0);
\draw[thick] (6,-1.5)  --(8.5,0);
\draw[thick] (6,-1.5)  --(3.5,0);
\draw[thick] (6,-1.5)  --(3.5,-3);
\draw[thick] (6,-1.5)  --(8.5,-3);
\end{tikzpicture}

    \caption{ Schemata of star concatenation with graphs $G_1, G_2 \cdots G_n$ from $L_1$.}
    \label{Sche_starcat}
\end{figure}
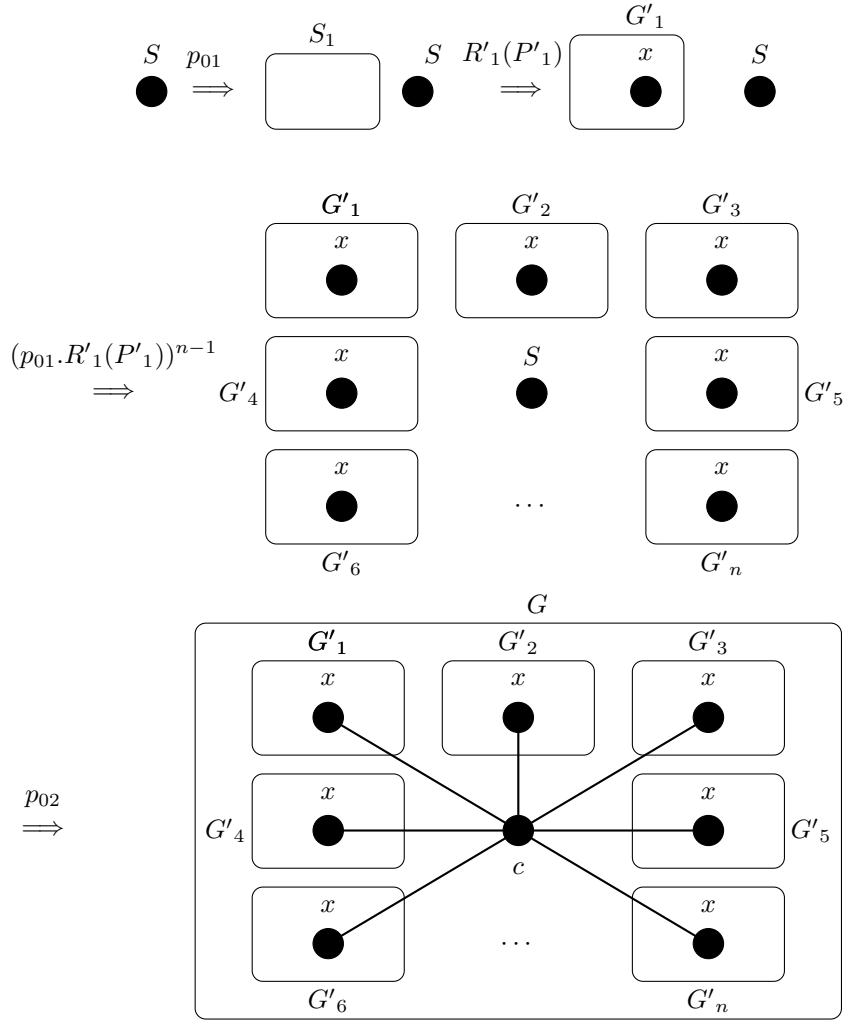
\section{Conclusion}\label{conclo}
In this paper we have proved  that $nc-eNCE$  graph grammars and its variants are closed under some simple graph operations. These closure properties could  be applied in situations where we wish to generate/study the classes of graph languages which could be recursively constructed from simpler families of graphs.  The  possibility for exploration into closure properties of these classes of grammars with respect to other operations remains open. Other properties such as the membership problem also  remain to be examined.
\bibliography{sn-bibliography}

\end{document}